\documentclass[runningheads]{llncs}

\usepackage{amsmath}
\usepackage{amssymb}
\usepackage{graphicx}
\usepackage{ucs}
\usepackage[utf8x]{inputenc}

\usepackage{algorithm}
\usepackage{algorithmic}
\usepackage{subcaption}
\usepackage{tikz}
\usepackage[outline]{contour}

\usepackage[T1]{fontenc}
\usepackage{amsmath,amssymb}
\usepackage{authblk}
\usepackage{booktabs}
\usepackage{hyperref}
\hypersetup{
    colorlinks,
    hidelinks,
    citecolor=blue,
    filecolor=black,
    linkcolor=red,
    urlcolor=green,
    plainpages=false,
    pdfpagelabels=true
}

\usepackage{url}
\newcommand{\keywords}[1]{\par\addvspace\baselineskip
\noindent\keywordname\enspace\ignorespaces#1}

\begin{document}

\mainmatter

\title{On the Algorithmic Recovering of Coefficients in Linearizable Differential Equations}
\titlerunning{Algorithmic Recovering of Coefficients in Linearizable Differential Equations}

\author{Dmitry~A.~Lyakhov\inst{1} \and 
Dominik~L.~Michels\inst{2}}
\institute{KAUST Computational Sciences Group, Campus Bldg.~1, Thuwal 23955-6900, KSA\\
\inst{1}\email{dmitry.lyakhov@kaust.edu.sa}\\
\inst{2}\email{dominik.michels@kaust.edu.sa}\\
}

\authorrunning{D.~A.~Lyakhov et al.}
\toctitle{Lecture Notes in Computer Science}
\tocauthor{D. Lyakhov et al.}

%%%%%%%%%%%%%%%%%%%%%%%%%%%%%%%%%%%%%%%%%%%%%%%%%%

\maketitle

\begin{abstract}
We investigate the problem of recovering coefficients in scalar nonlinear ordinary differential equations that can be exactly linearized. This contribution builds upon prior work by Lyakhov, Gerdt, and Michels, which focused on obtaining a linearizability certificate through point transformations. Our focus is on quasi-linear equations, specifically those solved for the highest derivative with a rational dependence on the variables involved. Our novel algorithm for coefficient recovery relies on basic operations on Lie algebras, such as computing the derived algebra and the dimension of the symmetry algebra. This algorithmic approach is efficient, although finding the linearization transformation necessitates computing at least one solution of the corresponding Bluman-Kumei equation system.

\end{abstract}

%%%%%%%%%%%%%%%%%%%%%%%%%%%%%%%%%%%%%%%%%%%%%%%%%%
\section{Introduction}
\label{sec:1}

Differential equations are essential tools in various scientific and engineering fields, enabling us to describe and understand a wide range of phenomena ranging from the dynamics of molecular structures~\cite{MICHELS2015336} and quantum systems~\cite{PhysRevA.101.043834} over electromagnetic effects~\cite{10.1145/3414685.3417799} to fluid mechanics~\cite{10.1145/3478513.3480495} and complex atmospheric phenomena~\cite{10.1145/3414685.3417801}.

In this paper, we present an algorithm aimed at transforming a non-linear differential equation into a linear form using an analytical invertible transformation. We introduce novel methods and theorems that surpass the current state-of-the-art in this area.

More than 150 years ago, Lie introduced a systematic approach that is now known as symmetry analysis. This powerful technique was later revitalized in the 20th century by Ovsyannikov and his students. Lie's vision involved using groups of point transformations to tackle systems of partial differential equations (PDE). His ingenious idea was first to identify the infinitesimal generators of one-parameter symmetry subgroups, which would then be used to build the complete symmetry group. By analyzing the symmetries within a differential equation, we can obtain valuable insights into the underlying structure of the problem it represents. Notably, the presence of a symmetry group offers several advantages: It can reduce the differential order of an equation, potentially transform it from partial to ordinary form, and facilitate the construction of particular and even, in some cases, general solutions.

In contrast to Lie, Lyakhov et al.~recently discovered~\cite{lyakhov2017algorithmic,lyakhov2020algorithmic} that such kind of properties as exact linearization could be detected algorithmically without solving the determining system. It relies strongly on differential algebra and symbolic manipulations with differential equations. Once the linearization certificate is obtained for the differential equation under consideration, the next step is constructing the linearizing mapping. In previous work~\cite{lyakhov2017algorithmic,lyakhov2020algorithmic}, we used \textbf{LinearizationTestII} for the construction of the determining system. Unfortunately, it suffers from substantial computational complexity caused by the high non-linearity of the initial PDE system. In this regard, it is worth to mention a novel approach \cite{mohammadi2019introduction,mohammadi2021symmetry} working on the a ``linearized'' Lie algebra level instead of the nonlinear Lie Group. Linear structures of the same dimension is in favor in computational algebra as they satisfy complexity requirements.

%This paper is organized as follows. In Sect.~\ref{sec:2}, we briefly describe the mathematical objects we deal with and the former result on linearization by point transformation~\cite{LyakhovGerdtMichels}. In Sect.~\ref{sec:3}, we introduce contact symmetry and prove the main theorem of our paper. The implementation of algorithms and its application is illustrated in Sect.~\ref{sec:4} by several examples. Finally, we provide a conclusion in Sect.~\ref{sec:5}.

\section{Linearization Certificate}
\label{sec:2}

Let us consider scalar ordinary differential equations (ODE) of the form\\[-0.3cm]
\begin{equation}\label{ode}
y^{(n)}(x)+f(x,y,y^{\prime},\ldots,y^{(n-1)})=0\,,\quad y^{(k)}:=\frac{d^ky}{dx^k}
\end{equation}
with rational right-hand side $f\in \mathbb{C}(x,y,y^{\prime},\ldots,y^{(n-1)})$  solved with respect to the highest order derivative. Given an ODE of the form of Eq.~\eqref{ode}, we aim to check the
existence of an invertible holomorphic transformation
\begin{equation}\label{transformation}
u=\psi(x,y)\,,\quad t=\phi(x,y)\,,
\end{equation}
which maps Eq.~\eqref{ode} into a linear $n$-th order homogeneous equation
\begin{equation}\label{lode}
u^{(n)}(t)+\sum_{k=0}^{n-1} a_{k}(t)\,u^{(k)}(t)=0\,,\quad u^{(k)}:=\frac{d^ku}{dt^k}\,.
\end{equation}
The inequation provides local invertibility of diffeomorphism $(x,y) \rightarrow (t,u)$,
\begin{equation}\label{Jacobian}
\mathsf{det}(J) \neq 0\,,
\end{equation}
where
$$
J := \begin{pmatrix} \phi_x & \phi_y \\ \psi_x & \psi_y \end{pmatrix}.
$$

The classical approach of Lie~\cite{lie1893vorlesungen} consists of considering a group of (symmetry) transformations that map a set of solutions into itself. In this regard, linear equations admit a rich group of symmetries that we use to identify linearizability. A powerful way to study the symmetry properties of Eq.~\eqref{ode} is to consider the tangent space with {\em infinitesimal} transformation generators defined by the Taylor series in the vicinity of the identity mapping:
\begin{equation}\label{inf_trans}
\tilde{x}=x+ \varepsilon\,\xi(x,y) + \mathcal{O}(\varepsilon^2)\,,\ \ \tilde{y}=y + \varepsilon\,\eta(x,y) +
\mathcal{O}(\varepsilon^2)\,.
\end{equation}
Then {\em invariance condition} for Eq.~\eqref{ode} under the transformation given by Eq.~\eqref{inf_trans} could also be obtained by the differential equality
\begin{equation}\label{invariance}
 {\cal{X}}(y^{(n)}+f(x,y,\ldots,y^{(n-1)}))|_{y^{(n)}+f\left(x,y,\ldots,y^{(n-1)}\right)=0}=0\,,
\end{equation}
where the {\em symmetry operator} reads
\begin{equation}\label{symm_generator}
{\cal{X}}:=\xi\,{\partial_x}+\sum_{k=0}^n\eta^{(k)}\partial_{y^{(k)}}\,,\
\eta^{(k)}:=D_x\eta^{(k-1)}-y^{(k)}D_x\,\xi\,,
\end{equation}
$\eta^{(0)}:=\eta$, and $D_x:=\partial_x+\sum_{k\geq 0}y^{(k+1)}\partial_{y^{(k)}}$
is the total derivative operator with respect to $x$.

For a given ODE with rational right-hand side $f$, the invariance condition Eq.~\eqref{invariance} means that its left-hand side vanishes when
Eq.~\eqref{ode} holds. Then the application of Eq.~\eqref{symm_generator} to the left-hand side of Eq.~\eqref{ode}
and the substitution of $y^{(n)}$ with $-f(x,y,\ldots,y^{(n-1)})$ in the resulting expression leads to the equality $g=0$
with the polynomial dependence of $g$ on the derivatives $y^\prime,\ldots,y^{(n-1)}$. Since, by
Eq.~\eqref{inf_trans}, the functions  $\xi$ and $\eta$ do not depend on these derivatives, the equality $g=0$
holds if and only if all coefficients in $y^\prime,\ldots,y^{(n-1)}$ are equal to zero. This leads to an
overdetermined linear PDE system in $\xi$ and $\eta$ called {\em determining system} which is the main object for symbolic analysis.

In differential algebra, linear partial differential equations form an ideal in the ring of differential polynomials. At the beginning of the 20th century, Maurice Janet introduced a method of completing it to a normal form. For a given ranking, the algorithm yields an output system equivalent to the input one and initial data, ensuring the existence and uniqueness of a formal power series solution. Initial data determines the dimension of the solution space. In the case of constant coefficients, a linear differential system could be transformed into a polynomial equations system, and the Janet's basis is just a specific case of the Groebner basis.

The set of infinitesimal symmetry operators is a {\em Lie algebra} with Lie bracket $[\cdot,\cdot]$. If it is a finite-dimensional linear space, we denote its dimension by $m$.  Any set of $m$ linearly independent vectors forms a vector space basis. All possible Lie brackets could be expressed by 
\begin{equation}\label{LieAlgebra}
[{\cal{X}}_i,{\cal{X}}_j]=\sum_{k=1}^{m} C^k_{i,j}{\cal{X}}_k\,,\quad 1\leq i<j\leq m\,,
\end{equation}
where $C^k_{i,j}$ are important characteristics which are denoted as the {\em structure constants} of the Lie algebra.

\begin{theorem}
The point symmetry algebra of the scalar ODE given by Eq.~\eqref{ode} is always finite-dimensional if $n > 1$.
\end{theorem}

\begin{proof}Assuming the right-hand side of the second-order ODE is an analytical function, let us expand it in the Taylor series
$$f(x,y,y') = \sum_{i=0}^{\infty} f_i(x,y) (y')^i\,.$$
A formal procedure based on the previous proposition leads to an infinite number of equations, where the most important ones are the first four equations:
\begin{equation}
\begin{gathered}
\frac{\partial^2 \eta}{\partial x^2} = F_1(\xi_x, \xi_y, \eta_x, \eta_y, \xi, \eta)\,,\,\, -\frac{\partial^2 \xi}{\partial x^2} + 2 \frac{\partial^2 \eta}{\partial x \partial y} = F_2(\xi_x, \xi_y, \eta_x, \eta_y, \xi, \eta)\,,\\
\frac{\partial^2 \xi}{\partial y^2} = F_3(\xi_x, \xi_y, \eta_x, \eta_y, \xi, \eta)\,, \,\,-\frac{\partial^2 \eta}{\partial y^2} + 2 \frac{\partial^2 \xi}{\partial x \partial y} = F_4(\xi_x, \xi_y, \eta_x, \eta_y, \xi, \eta)\,.
\end{gathered}
\end{equation}
Differentiation by $x$ and $y$, and linear combination leads to all possible third-order derivatives of $\xi$ and $\eta$:
\begin{equation}
\begin{gathered}
(\xi, \eta)_{xxx} = (\bar{F}_1, \bar{G}_1)\,,\,\, (\xi, \eta)_{xxy} = (\bar{F}_2, \bar{G}_2)\,,\\ (\xi, \eta)_{xyy} = (\bar{F}_3, \bar{G}_3)\,,\,\, (\xi, \eta)_{yyy} = (\bar{F}_4, \bar{G}_4)\,.
\end{gathered}
\end{equation}
It is a finite-dimensional system with a maximal dimension of 12. Adding four constraints implies that the dimension of the solution space is at most 8 (it is possible to achieve as shown in Example~\ref{example:one}). For higher orders, the proof is similar; we refer to previous work~\cite{ovsiannikov2014group}.
\end{proof}

A nice property of the finite-dimensional Lie symmetry algebra is that the structure constants table can be found exactly without any heuristics. Using differential-elimination algorithms, we can complete it to the involutive form so that we will know the dimension of its solution set. The Taylor series of $m$ linearly independent vectors ${\cal{X}}_i$ being substituted into Eq.~\eqref{LieAlgebra} leads to the infinite system of linear equations for a finite number of coefficients $C^k_{i,j}$. This system is always equivalent to some truncated version, which leads to an efficient procedure for obtaining the structure coefficients \cite{reid1991finding}.

We denote the Lie symmetry algebra by $L$ and set $m=\dim(L)$.  Its {\em derived algebra}  $\mathcal{D}\subset L$
is a subalgebra that consists of all commutators of pairs of elements in $L$.

\begin{theorem}\cite{lyakhov2017algorithmic} Scalar differential Eq.~(\ref{ode}) with $n\geq 2$ is linearizable if and only if one of the following conditions is fulfilled:
\begin{enumerate}
\item $n=2\,,m=8$ or $n\geq 3\,,m=n+4$ (maximality of dimension space);
\item $n\geq 3\,,m\in \{n+1,n+2\}$ and derived algebra $\mathcal{D}$ is abelian of dimension $n$.
\end{enumerate}
\end{theorem}

\begin{remark}\label{symmetrycaseslinearODE}
Symmetry dimension maximality implies equivalence to trivial the equation $u^{(n)}(t) = 0$, $m = n+2$ corresponds to the constant coefficients case for some linearizing mapping, and $m = n+1$ matches principally non-constant coefficient linearizability. 
\end{remark}

Solutions of differential equations (partial or ordinary) are defined in differential algebra as formal power series. The natural question is whether they are convergent. Riquier has shown that this is correct if an appropriate ranking is used. 

\begin{definition} A ranking $R$ is a Riquier ranking if $D_1 u > D_2 u$ (where $u$ is differential indeterminate and $D_1, D_2$ are differential operators) implies $D_1 v > D_2 v$ for all differential indeterminates $v$.
\end{definition}

Riquier rankings are quite restrictive in differential algebra, but there are practical algorithms to construct them \cite{rust1997rankings}. As Riquier proved in his analyticity theorem (a generalization of the Cauchy–Kovalevskaya theorem), the unique formal power series defined by initial data in the Riquier ranking is converging. It justifies the formal certificate obtained by \textbf{LinearizationTestI}~\cite{lyakhov2017algorithmic,lyakhov2020algorithmic}. Widely used in computer algebra, the orderly ranking could not be sufficient for convergence of a formal power series solution as shown by Lemaire's elementary example \cite{lemaire2003orderly}.

%%%%%%%%%%%%%%%%%%%%%%%%%%%%%%%%%%%%%%%%%%%%%%%%%%
\section{Linear Equations with Constant Coefficients}
\label{sec:3}

In this chapter, we focus on the special case $m = n+2$ which corresponds to linearizability with constant coefficients (non-trivial case). We will show that for this case, we can algorithmically find the form of equation just by simple manipulations with the abstract Lie algebra of symmetries.

Let firstly remind the reader that there exists a one-to-one correspondence between the linear homogeneous scalar ordinary differential equation of $n$-th order and its characteristic polynomial
$$
u^{(n)}(t)+\sum_{k=0}^{n-1} a_k\,u^{(k)}(t)=0\, \leftrightarrow [u = e^{z t}] \leftrightarrow\, f(z) = \prod^{n}_{i=1}(z-\lambda_i)\,,
$$
thus without limitation of generality we will work mostly with characteristic polynomials.

The transformation
\begin{equation}\label{eq-trans}
\bar{t}=\frac{t}{k},\,\, \bar{u}=u \, \mathsf{exp} \left(\frac{bt}{k} \right)
\end{equation}
defines an isomorphism in the space of linear homogeneous differential equations with constant coefficients. Obviously, it maps 
\begin{equation}\label{arbitrariness}
f_1(z)=\prod^{n}_{i=1}(z - \lambda_i) \rightarrow f_2(z)=\prod^{n}_{i=1}(z - k\bar{\lambda}_i - b)\,.
\end{equation}
\begin{theorem}\label{th-kb}
Equivalence transformations of linear equation with constant coefficients define arbitrariness of characteristic polynomial up to values of $k, b$ in~(\ref{arbitrariness}).
\end{theorem}

\begin{example}\label{example:one}
Let us consider firstly a simple but conceptual example. The characteristic polynomial has two equal eigenvalues $\lambda_1$ and $\lambda_2$. It means that the fundamental solution of the corresponding ODE consists of 
\begin{equation}
y_1(x) = e^{\lambda_1 x}\,, \,\,y_2(x) = x e^{\lambda_1 x}\,, \,\,y_3(x) = e^{\lambda_2 x}\,.
\end{equation}
Then, the derived algebra $D$ of the symmetry algebra is defined by the linear span of three operators:
\begin{equation}
\left\{e^{\lambda_1 x} \frac{\partial}{\partial y}\,, \,\,x e^{\lambda_1 x} \frac{\partial}{\partial y}\,, \,\,e^{\lambda_2 x} \frac{\partial}{\partial y} \right\}\,.
\end{equation}
The operator
$$
Y = \frac{\partial}{\partial x}
$$
defines a left action of the Lie algebra $(Y, \cdot)$. It maps elements from the derived algebra into itself ($y_i \rightarrow w_i$). Then, we can expand new elements $w_i$ in a basis $y_i$ of the linear space and immediately obtain

\begin{equation}
\begin{bmatrix}
  \lambda_1 & 0 & 0 \\
  1 & \lambda_1 & 0 \\
  0 & 0 & \lambda_2
\end{bmatrix}
\begin{bmatrix}
  y_1 \\
  y_2 \\
  y_3
\end{bmatrix}
=
\begin{bmatrix}
  w_1 \\
  w_2 \\
  w_3
\end{bmatrix}.
\end{equation}

An important point is that the matrix of this expansion $A_{ij}: \textbf{A} \cdot y = w$ shares the same characteristic polynomial $f(z)$ as the initial ODE. Further, let us suppose we have basis vectors $\tilde{y}_i$ which are not individual exponents. Then, we can expand them using the basis of $y_i = T \tilde{y}_i, \,\mathsf{det}(T) \neq 0$:

\begin{equation}
\begin{bmatrix}
  \lambda_1 & 0 & 0 \\
  1 & \lambda_1 & 0 \\
  0 & 0 & \lambda_2
\end{bmatrix}
T
\begin{bmatrix}
  \tilde{y}_1 \\
  \tilde{y}_2 \\
  \tilde{y}_3
\end{bmatrix}
= T
\begin{bmatrix}
  \tilde{w}_1 \\
  \tilde{w}_2 \\
  \tilde{w}_3
\end{bmatrix}.
\end{equation}
Obviously, equivalent matrices
\begin{equation}
T^{-1}
\begin{bmatrix}
  \lambda_1 & 0 & 0 \\
  1 & \lambda_1 & 0 \\
  0 & 0 & \lambda_2
\end{bmatrix}
T \sim
\begin{bmatrix}
  \lambda_1 & 0 & 0 \\
  1 & \lambda_1 & 0 \\
  0 & 0 & \lambda_2
\end{bmatrix}
\end{equation}
share the same characteristic polynomial. This key observation is crucial for our coefficient recovery algorithm.
\end{example}

In the general situation of linearizable equation, we cannot distinguish operator $Y$, but instead, we can choose some non-trivial operator that does not belong to derived algebra. This operator could be described as
$\tilde{Y} = k Y + b I$, where $I$ is the identity operator. Repeating the whole procedure with operator $\tilde{Y}$ we get characteristic polynomial
$$
f(z) = (z - k \lambda_1 - b)^2 (k - k \lambda_2 - b)\,,
$$
which differs from the original characteristic polynomial by linear mapping of eigenvalues. But according to Theorem \ref{th-kb}, it corresponds to the characteristic polynomial of linear equation with constant coefficients in the same equivalence class.

%%%%%%%%%%%%%%%%%%%%%%%%%%%%%%%%%%%%%%%%%%%%%%%%%%
\section{Algorithmics}
\label{sec:4}

The described Example \ref{example:one} could be easily generalized to the general coefficient recovery method. In the general case, matrix $A_{ij}$ will contain several Jordan blocks. Without any principal changes, we can apply the above technique which leads to the following Algorithm~\ref{alg3}.

\begin{algorithm}
\caption{{\textsl{\bfseries{Coefficients Recovery ($m=n+2$).}}}
\label{alg3}}
{\bf Input:} $q$, a linearizable ODE in the form of Eq.~\eqref{ode} with $m=n+2$.\\
{\bf Output:} $E$, a linear ODE with constant coefficients which is equivalent to $q$.

\begin{algorithmic}[1]
\STATE Compute the {\textsl{\bfseries{Determining System}}} of the symmetry group infinitesimals.
\STATE Extract the abstract {\textsl{\bfseries{Lie Algebra}}} by means of structure constants.
\STATE Find {\textsl{\bfseries{Derived Algebra}}} by computing Lie brackets between all basis vectors.
\STATE Find {\textsl{\bfseries{Factor Space}}} of Lie algebra by {\textsl{\bfseries{Derived Algebra}}}: ({\textsl{\bfseries{F}}}, $\{e_1, e_2\}$).
\STATE Define non-trivial left action $\left[e_1\,, \,\cdot\right]$ or $\left[e_2\,, \,\cdot\right]$.
\STATE Form a matrix $A$ defined by action on basis vectors in {\textsl{\bfseries{Derived Algebra}}}.
\STATE Find the characteristic polynomial of this matrix.
\end{algorithmic}
\end{algorithm}

%%%%%%%%%%%%%%%%%%%%%%%%%%%%%%%%%%%%%%%%%%%%%%%%%%
\section{Conclusison}
\label{sec:5}

We have successfully constructed a new algorithm which helps to recover coefficients in a case of constant coefficients linearizability. The algorithm enhances \textbf{LinearizationTestII} (and \textbf{Algorithm 3} MapDE with Target$=$LinearDE) as knowledge of the form of equations decreases complexity of the algorithms based on differential algebra. In the case of $m = n+1$, the Lie algebra looks trivial and it is not possible to recover principally non-constant coefficients just by manipulations with the abstract Lie algebra of symmetries.

%%%%%%%%%%%%%%%%%%%%%%%%%%%%%%%%%%%%%%%%%%%%%%%%%%
\section*{Acknowledgements}

This work has been funded by the baseline funding of the KAUST Computational Sciences Group. The authors thank Yang Liu for useful discussions.
%{\color{blue}The reviewers' valuable comments that improved the manuscript are gratefully acknowledged.}

\bibliographystyle{unsrt}
\bibliography{main}
%%%%%%%%%%%%%%%%%%%%%%%%%%%%%%%%%%%%%%%%%%%%%%%%%%
\end{document}